\newcommand{\columnsversion}[2]{#1}{} %one-column version

\columnsversion{
\documentclass[11pt,letterpaper]{article}

\topmargin=-0.35in \topskip=0pt \headsep=15pt \oddsidemargin=0pt
\textheight=9in \textwidth=6.52in \voffset=0in

}{
\documentclass{sig-alternate}
}

\usepackage{amsmath,amsfonts,graphpap,amscd,mathrsfs,graphicx,lscape}
\usepackage{epsfig,amssymb,amstext,xspace}
\usepackage{theorem, bbm}

\usepackage[ruled]{algorithm2e}
\renewcommand{\algorithmcfname}{ALGORITHM}
\SetAlFnt{\small}
\SetAlCapFnt{\small}
\SetAlCapNameFnt{\small}
\SetAlCapHSkip{0pt}
\IncMargin{-\parindent}

\usepackage{color}              % Need the color package
\usepackage{epsfig}
\usepackage{paralist}
\usepackage{multirow}

\usepackage{caption}
\usepackage{url}
\usepackage[normalem]{ulem}
\usepackage{color-edits}

\addauthor{rpl}{red}
\addauthor{gg}{blue}
\addauthor{vm}{green}

\newtheorem{theorem}{Theorem}[section]
\newtheorem{lemma}[theorem]{Lemma}

\newtheorem{cor}[theorem]{Corollary}
\newtheorem{assumption}[theorem]{Assumption}

\newtheorem{example}[theorem]{Example}
\newtheorem{defn}[theorem]{Definition}

\newcommand{\algspace}{\text{    } \quad  \text{    }}

{\theorembodyfont{\rmfamily} }
\def\squareforqed{\hbox{\rule{2.5mm}{2.5mm}}}

\def\QED{\ifmmode\squareforqed % in mathmode : print just the square
  \else{\nobreak\hfil   % \hfil to end of current line
    \penalty50                 % penalty 50 for breaking the line here
    \hskip1em                  % leave at least 1em before the square
    \null                      % \hbox{}
    \nobreak                   % prohibit line break
    \hfil                      % another \hfil (if a break occurred)
    \squareforqed              % put the square here
    \parfillskip=0pt           % the line really ends here
    \finalhyphendemerits=0     % ignore a hyphen on the line above
    \endgraf}                  % end the paragraph
  \fi}

\def\blksquare{\rule{2mm}{2mm}}
\def\qedsymbol{\blksquare}
\newcommand{\bg}[1]{\medskip\noindent{\bf #1}}
\newcommand{\ed}{{\hfill\qedsymbol}\medskip}

\columnsversion{
\newenvironment{proof}{\bg{Proof : }}{\ed}

}{}
\newenvironment{proofof}[1]{\bg{Proof of #1 : }}{\ed}

\newcommand{\R}{\ensuremath{\mathbb R}}

\newcommand{\comment}[1]{}
{}  %Alternative to be used for 10 page
\newcommand{\junk}[1]{}

\newlength{\tmp} \newlength{\lpsx} \newlength{\lpsy} \newlength{\upsx}
\newlength{\upsy}

\columnsversion{
\newcommand{\email}[1]{\texttt{#1}}
}{}

\newcommand{\one}{\mathbbm{1}}

\newcommand{\A}{\ensuremath{\mathcal{A}}}

\newcommand{\checkpoint}[1]{\begin{picture}(0,0)
\put(10,-7){ $\longleftarrow$ point \textbf{#1}}
\end{picture}}

\newcommand{\hi}{{\hat{i}}}
\newcommand{\cp}{{ ($\clubsuit$) }}

\begin{document}

\setcounter{page}{0}

\title{Clinching Auctions Beyond Hard Budget Constraints}
%\title{Clinching Auctions, Budgets and Online Supply}
%\title{Ascending Auctions, Budgets and Online Supply}

\author{
Gagan Goel\\
       Google Inc., New York\\
       \email{gagangoel@google.com}
% 2nd. author
\and
Vahab Mirrokni\\
       Google Inc., New York\\
       \email{mirrokni@google.com}
% 2nd. author
\and
Renato Paes Leme\\
       Google Inc., New York\\
        \email{renatoppl@google.com}
}

\date{}

\maketitle

\begin{abstract}

Constraints on agent's ability to pay play a major role in auction design for
any setting where the magnitude of financial transactions is sufficiently large.
Those constraints have been traditionally modeled in mechanism design as
\emph{hard budget}, i.e., mechanism is not allowed to 
charge agents more than a certain amount. Yet, real auction systems (such as
Google AdWords) allow more sophisticated constraints on agents' ability to pay,
such as \emph{average budgets}. In this work, we investigate the
design of Pareto optimal and incentive compatible auctions for agents with
\emph{constrained quasi-linear utilities}, which captures more realistic models
of liquidity constraints that the agents may have. Our result applies to a very
general class of allocation constraints 
known as polymatroidal environments, encompassing many settings of interest such
as multi-unit auctions, matching markets, video-on-demand and advertisement
systems. 

Our design is based Ausubel's \emph{clinching framework}. Incentive
compatibility and feasibility with respect to ability-to-pay constraints are
direct consequences of the clinching framework. Pareto-optimality, on the other
hand, is considerably more challenging, since the no-trade condition that
characterizes it depends not only on whether agents have their budgets exhausted
or not, 
but also on prices {at} which the goods are allocated. In order to get a handle
on those prices, we introduce novel concepts of
dropping prices and saturation. These concepts lead to our main structural
result which is a characterization of the tight sets in the 
clinching auction outcome and its relation to dropping prices.

\comment{
Recent progress in algorithmic mechanism design for non-quasi-linear settings
has been mostly restricted to hard budget constraints,
where each agent has a \emph{hard budget} and the mechanism is not allowed to
charge him more than that. Yet, real auction systems (such as Google AdWords)
 allow more sophisticated constraints on agents' ability to pay.
In this work, we bridge this gap by designing auctions for players with
\emph{constrained quasi-linear utilities}, which captures more realistic models
of liquidity constraints that the agents may have. Our result applies to a very
general class of allocation constraints
known as polymatroidal environments, encompassing many settings of interest such as multi-unit
auctions, matching markets, video-on-demand and advertisement systems.

The auction itself is a simple twist from the Polyhedral Clinching Auction in Goel, Mirrokni and Paes Leme (STOC'12).
Proving Pareto-optimality in this more general setting, however, is considerably more challenging, since the no-trade
conditions that characterize Pareto-optimality depend not only on whether agents have their budgets exhausted or not,
but also on prices {at} which the goods are allocated. In order to get a handle
on those prices, we introduce novel concepts of
dropping prices and saturation. These concepts lead to our main structural result which is a characterization of the tight sets in the
clinching auction outcome and its relation to dropping prices.}
\end{abstract}

\newpage

\section{Introduction}
An important direction in  mechanism design is to understand how to
design efficient mechanisms when players have constraints {on} their ability to
pay. A first order approximation is to consider \emph{hard budget}
constraints, in which each agent has a budget and the mechanism is not allowed
to charge him more than this amount. While simpler and more theoretically
tractable, hard budgets stand usually as a proxy for more sophisticated
payment constraints.

A recent trend in modern internet marketplaces such as Google AdWords is to
offer the bidders a better control of their spending by allowing them to
express more sophisticated constraints on their ability to pay.
{A popular feature introduced by Google Adwords in 2010 called ``Target CPA
bidding'' allows advertisers to report \emph{average budget} constraints on top
of traditional willingness to pay per item (value) and hard budgets (see
\cite{adwords_blog} for a discussion of this feature in the Google AdWords
blog).}

%A popular feature introduced in 2010 is called ``Target CPA bidding'' (see
%\cite{adwords_blog} for a discussion of this feature in the Google AdWords
%blog), in which advertisers can report \emph{average budget} constraints on top
%of traditional willingness to pay per item (value) and hard budgets.

It is important to emphasize that values, hard budgets and average budgets play
different roles in managing an advertising campaign, or more generally satisfying
buyers' desired goals.
 In order to illustrate this point, consider a marketplace
in which each agent specifies his preferences and gets allocated a certain
quantity of a good (ad impressions, for example) and charged a total amount for
it. Hard budgets are one of  the simplest constraints on the total payment: they specify
an upper bound on the total payment. Average budgets specify  an upper bound on
the ratio of total payment by amount of goods allocated {(or alternatively, a
lower bound on the ROI, return over investment)}.
On the other hand,  individual valuations specify an upper bound on the marginal
payment for each individual item, even if some goods are sold at a lower or higher
price earlier. To see the difference more clearly, consider an initial
outcome where an agent gets some items and pays a certain amount that is below
his average budget. If he is offered an extra item for a price less then his
value but higher than the average budget, he would prefer the outcome with the
extra item as long as the new total payment and allocations don't exceed his
average or hard budget constraints. A natural generalization of average budget
constraint is to consider a concave upper bound on the total payment as a function
of the number of goods allocated.\\

We consider here the problem of designing Pareto-optimal mechanisms for settings
where the players have \emph{general} (concave) constraints on their ability to
pay. This includes hard budgets and average budgets (and combinations thereof),
as well as other more sophisticated constraints on the total payment of agents
as a function of the set of goods allocated to them. For the special class of
hard budgets, a sequence of papers \cite{dobzinski12, fiat_clinching,
henzinger11,goel12, goel13, devanur12} studied this problem
for increasingly complex classes of {\em allocation environments} using Ausubel's
celebrated clinching framework \cite{Ausubel_multi} as the main tool.
Nonetheless, all those results are restricted to hard budgets, and do not handle
more general payment constraints. \comment{To the best of our knowledge, our work is the first result that
goes beyond the simple hard budget constraints, even for the case of multi-unit
auctions.}

\paragraph{Our results and techniques.}
In this paper, we study the constrained quasi-linear model, in which
each agent has a private valuation, and has associated
with it a public\footnote{The assumption that the set of admissible outcomes is
public is necessary. Dobzinski et al \cite{dobzinski12} showed that even for
the special case of multi-unit auctions with hard budget constraints, there is
no incentive compatible, individually rational and Pareto-optimal auction if
budgets are private. Indeed, most papers in the literature on budgets make the
public budgets assumption \cite{dobzinski12, fiat_clinching, henzinger11,
goel12, goel13, devanur12}, including classical references such as Laffont and
Roberts \cite{laffont_robert} and Maskin \cite{maskin00}. }
set of \emph{admissible outcomes} where each outcome is a pair
of allocation and payments. The utilities are then quasi-linear if the outcome
is admissible and minus infinity otherwise.
Our main result is to design an incentive compatible, individually rational and Pareto-efficient auction
that handles a general class of these payment constraints.
\comment{This is first result of this kind that goes beyond the simple hard
budget constraints.} It is worth noting that many attempts to generalize
clinching auctions to other settings such as private budgets \cite{dobzinski12},
or single-parameter concave valuations \cite{goel12} have led to impossibility
results. In light of that, it is somehow surprising to find a extension for
dealing with a general set of payment constraints for which clinching auctions
are flexible, and we get a positive result.

\comment{
For the special class of hard
budget constraints, Dobzinski, Lavi and Nisan \citeyear{dobzinski12} generalized
the celebrated clinching framework  of Ausubel~\citeyear{Ausubel_multi}
and presented the first such auctions for the special class of multi-unit auctions.
Following this paper, there has been a sequence of results presenting such
auctions for more general classes of allocation environments\cite{fiat_clinching, henzinger11,
goel12, goel13, devanur12}; such as matching constraints and polymatroidal constraints.
% The most general set of allocation constraints are polymatroidal environments.
Nonetheless, all the previous results apply only to the
special class of hard budgets, and does not handle more general  payment constraints.
%presenting such auctions
}

Our result applies to a very general class of allocation constraints
known as polymatroidal environments. Polymatroidal
environments encompass many settings  such as multi-unit auctions, advertisement systems,
matching markets, video on demand (routing), and spanning tree auctions. See
Goel et al~\cite{goel12}
and Bikhchandani et al~\cite{Bikhchandani11}  for a more comprehensive
discussion on applications of polymatroidal environments.

Algorithmically, our auction can be thought of as a variant of the polyhedral clinching auction
in Goel, Mirrokni and Paes Leme \cite{goel12} with a more general
demand function. While applying a variant
of the previously known algorithm, proving that the auction is Pareto-optimal for
more sophisticated payment constraints becomes considerably more challenging,
and requires novel techniques.
The reason is as follows: Pareto optimality is usually characterized by a no-trade
condition, which states that given two players $H$ and $L$ with values
$v_H > v_L$, then for any price $p$, it should not be possible to take goods away
from the $L$ player by paying him at a rate $p$ for the goods taken away and allocating them to the $H$ player
charging him a rate $p$, such that it improves the utility of one of the them without
making the other worse-off. For hard budgets, there are two obstructions to trade:
either $H$ has his budget exhausted in the final solution or the trade violates
the  allocation feasibility constraints, i.e., $H$ is receiving goods to his
maximum capacity. Since neither obstruction depends on the specific price,
one can show that if it is possible to trade at price $p \in [v_L, v_H]$, then it is also possible
to trade at price $p = v_H$, and vice versa. This implies that one needs to check for no-trade
at price $p = v_H$ only, which greatly simplifies the analysis. Now for more general constraints
 (such as average budget constraint), this is not true. Meaning it might be possible to trade at
a price $p$ that is strictly between $v_L$ and $v_H$ but not at price $p = v_H$.
The harder part of our analysis is to get a handle on these prices. In order to do so, we define
the concept of \emph{dropping prices} which serve as an upper bound on the prices for
which trade is possible. We then relate those dropping prices to the feasibility constraints. This
leads to our main structural result which is encapsulated in our \emph{Structure of Tight
Sets Lemma} (Lemma \ref{lemma:structure_tight_sets} and its Corollary
\ref{cor:structure}). We believe that that this lemma exposes an interesting
structure about clinching auctions that can lead to other applications.\\

We believe that an important contribution of our work is to show that the
clinching framework can be applied to general types of payment constraints. For
the special case of hard budgets, clinching has been recently used as a
building block to achieve a variety of objectives: Goel et al \cite{goel13}
use it to design online allocation rules, Devanur, Ha and Hartline
\cite{devanur12} use it as a building block to approximate revenue
in budgeted settings and Dobzinski and Paes Leme \cite{quantitative13}
use it to approximate an efficiency-related objective. We believe that the
ideas in this paper are a first step towards solving other problems (online
allocation rules and revenue extraction, for example) for more general types of
payment constraints.

\paragraph{Related work}
Auction design with constraints on player's ability to pay have been
extensively studied in the literature. Most of the work is devoted to
understand the impact of \emph{hard budget constraints} in standard auctions,
see Che and Gale \cite{che_gale} and Benoit and Krishna
\cite{benoit_krishna},
for example, or optimize the revenue in the presence of budget constraints, as
in Laffont and Roberts \cite{laffont_robert}, Borgs et al
\cite{Borgs05},
Chawla et al \cite{chawla11},  Malakhov and Vohra \cite{malakhov_vohra}
and
Pai and Vohra \cite{pai_vohra}.

The research line of designing Pareto-optimal incentive-compatible 
mechanisms with budget constraints was started by
Dobzinski, Lavi and Nisan \cite{dobzinski12}, who study agents with hard
budget constraints in a multi-unit auctions setting, i.e., there is a limited
supply of identical objects to be sold and the agents have additive valuation
over the objects. They also point out that traditional welfare maximization
is impossible in budgeted settings and establish Pareto-optimality as the
natural efficiency goal for settings with payment constraints. Those ideas were
extended in many different directions in subsequent work: Bhattacharya et al
\cite{Bhattacharya10} study the divisible case and propose budget
elicitation
schemes, Fiat et al \cite{fiat_clinching} study the same problem for
matching
markets, Colini-Baldeschi et al \cite{henzinger11} study the single-keyword
sponsored search setting and Goel et al \cite{goel12} propose a polyhedral
version of this auction that gives a Pareto-optimal auction with hard
budgets for any environment that can be modeled as a polymatroid. Such
settings include sponsored search, matching markets, and routing auctions.
They also prove the impossibility of designing incentive-compatible auctions
for very simple polyhedral domains, beyond polymatroidal environments.

For constraints beyond hard budgets, the design of Pareto efficient auctions
has been restricted to unit demand settings, where each agent can be allocated
at most one item. Aggarwal, Muthukrishnan, Pal and Pal
\cite{Aggarwal09}, Dutting, Henzinger and Weber \cite{dutting11}, Alaei,
Jain
and Malekian \cite{alaei10}, Morimoto and Serizawa
\cite{morimoto_serizawa} use
techniques inspired by the Walrasian equilibrium literature to design
Pareto-efficient auctions. Baisa \cite{Baisa} studied to which extent
efficiency
and revenue can be optimized with minimal assumptions on agent's utilities.

\section{Setting}

\subsection{Constrained quasi-linear utilities}

We consider a natural generalization of budget constrained utilities
which we call \emph{constrained quasi-linear utitilies}. In this utility
model, player $i$ is characterized by a private valuation $v_i$ and a public
\emph{admissible set} $\A_i \subseteq \R^2_+$. Upon getting $x_i$ units of a
divisible and
homogeneous good and paying $\pi_i$ dollars for it, we consider that player
$i$'s utility is:
$$u_i(x_i, \pi_i) = \left\{ \begin{aligned} & v_i \cdot x_i - \pi_i, & & (x_i,
\pi_i) \in
\A_i \\ & - \infty & & \text{otherwise} \end{aligned} \right.$$
In other words, a player behaves
like a quasi-linear player if his outcome $(x_i, \pi_i)$ is admissable and has
minus infinity utility otherwise. For example, budget-constrained utility
functions are characterized by $\A_i = \{(x_i, \pi_i) \in \R^2_+; \pi_i \leq B_i
\}$. Average budget constraints can be represented by $\A_i = \{(x_i, \pi_i) 
\in \R^2_+; \pi_i \leq \beta_i x_i\}$. Generally, we consider any admissible set
of the form $\A_i = \{(x_i, \pi_i) \in \R^2_+; \pi_i \leq \alpha_i(x_i)\}$ for
concave non-decreasing functions \emph{ability-to-pay function} $\alpha_i : \R_+
\rightarrow \R_+$ with $\alpha_i(0) = 0$. 

The condition $\alpha_i(0) = 0$ expresses that a player get zero utility for
the zero allocation and zero payments. The fact that $\alpha_i$ is
non-decreasing expresses that if a player considers a certain outcome 
admissible, it also considers admissible any outcome where he is allocated at
least as much and pays no more then the original outcome. Finally, the concavity
of $\alpha_i$ expresses that if an agent considers certain outcomes admissible,
it considers any distribution (convex combination) of such outcomes also
admissible.

\comment{
Those admissible sets are depicted in
Figure \ref{fig:admissible}.  We say that an admissible set $\A_i$ is
\emph{valid} if it satisfies the following natural properties:
\begin{enumerate}
 \item[(A)] \emph{admissibiliy of the null outcome}: $(0,0) \in \A_i$, i.e., any
player gets zero utility by
not getting any goods and not paying anything.
 \item[(B)] \emph{right/down-wards closeness}: if $(x_i, \pi_i) \in \A_i$
then $(x'_i, \pi'_i) \in \A_i$ for any
$x'_i \geq x_i$ and $\pi'_i \leq \pi_i$, i.e., if a player considers some
outcome admissible, then he also considers admissible outcomes where he gets
allocated no less and pays no more.
 \item[(C)] \emph{convexity}: if $(x_i, \pi_i), (x'_i, \pi'_i) \in
\A_i$ then $(t x_i + (1-t)
x'_i, t \pi_i + (1-t) \pi'_i) \in \A_i$ for any $t \in [0,1]$, i.e., if a player
finds a set of outcomes admissible, then it finds all distributions over such
outcomes also admissible.
 \item[(D)] \emph{topological closeness}: the set $\A_i$ is topologically
closed, i.e., if $(x_i^s, \pi_i^s)
\in \A_i$ are a sequence of admissible outcomes such that $(x_i^s, \pi_i^s)
\rightarrow (x_i^*, \pi_i^*)$, then $(x_i^*, \pi_i^*) \in \A_i$.
\end{enumerate}
}

\begin{figure}
\centering
\includegraphics[scale=.8]{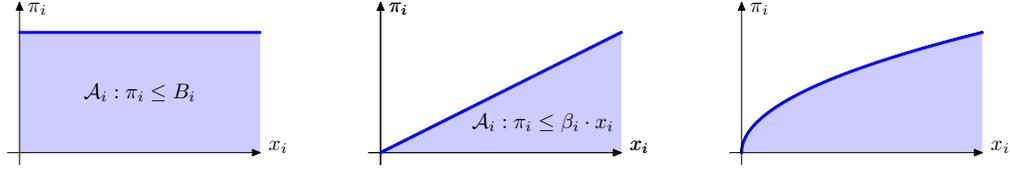}
\caption{Three examples of admissible regions: the first is an example
of \emph{hard budget constraints}, the second of \emph{average budget
constraints} and the third is an example of generic \emph{ability-to-pay}
function. }
\label{fig:admissible}
\end{figure}

\comment{
The set of conditions stated above imply that a valid admissible set can be
characterized by a concave \emph{ability-to-pay function} $\alpha_i(x_i) =
\max \{ \pi_i; (x_i, \pi_i) \in \A_i\}$. It is easy to see that the
conditions above imply that $\A_i = \{(x_i, \pi_i) \in \R^2_+; \pi_i \leq
\alpha_i(x_i)\}$.
}

\subsection{Polyhedral Auctions}

Given $n$ agents where each agent $i$ has a constrained quasi-linear utility
function with
value $v_i$ and a valid admissable set $\A_i$, our main problem is how to
auction a divible good that may be subject to allocation constraints. We assume
that the set of allocations is described by a convex set $P \subseteq \R^n_+$
such that a point $x = (x_1, \hdots, x_n) \in P$ if it is possible to
simultaneously allocate $x_i$ units to each player $i$. We call such a set the
\emph{environment}.

We assume that both the admissible sets $\A_i$ and the set of
feasible allocations $P$ is public information. The private information
of the agents is their value $v_i$ for each unit of the good. An auction is
described by two maps: 
$x:\R^n_+ \rightarrow P$ and $\pi: \R^n_+ \rightarrow \R^n_+$ such that for
each valuation profile $v \in \R^n_+$ it associates an allocation $x(v)$ and
payments $\pi(v)$.  Our goal is to design an auction that satisfied the
following properties:

\begin{itemize}
 \item \emph{admissibility}, i.e., $(x_i(v), \pi_i(v)) \in \A_i$
for each $i \in [n], v \in \R^n_+$.
 \item \emph{incentive compatibility} and
\emph{individual-rationality}: each player's
utility is maximized by reporting their true value regardless of the reports of
other agents. Moreover, he always gets non-negative utility by doing so. The
classic result by Myerson \cite{myerson-81} shows that this is equivalent to
$x_i(v_i, v_{-i})$ being monotone in $v_i$ and the payment rule be such that
$\pi_i(v_i, v_{-i}) = v_i \cdot x_i(v_i, v_{-i}) - \int_0^{v_i} x_i(u,
v_{-i}) du$.
 \item \emph{Pareto-efficiency}: we say that an outcome is
Pareto-efficient if there is no alternative outcome in which all the players'
utilities and auctioneer's revenue do not decrease and at least one increases.
Formally, an outcome $(x,\pi)$ is Pareto-efficient if there is no alternative
outcome $(x', \pi')$ with $x' \in P$, $(x'_i, \pi'_i) \in \A_i, \forall i$,
$v_i \cdot x'_i - \pi'_i \geq v_i \cdot x_i - \pi_i$, $\sum_i \pi'_i \geq \sum_i
\pi_i$ and the sum of those inequalities is strict, i.e., $\sum_i v_i \cdot
x'_i >  \sum_i v_i \cdot x_i$.
\end{itemize}

\subsection{Polymatroidal Environments}

Our results apply to the case where the allocation environment $P$ is a polymatroid. A
\emph{polymatroid} is a packing polytope that can be written as:
$$P = \{ x \in \R^n_+; \textstyle\sum_{i \in S} x_i \leq f(S), \forall S
\subseteq [n] \}$$
for a monotone submodular function $f: 2^{[n]} \rightarrow \R_+$. A
\emph{submodular function} is a set function that has the diminishing marginals
property: $f(S \cup i) - f(S) \geq f(T \cup i) - f(T)$ for any
subsets $S \subseteq T \subseteq [n]$. An equivalent (and somewhat more
traditional) definition is: $f(S \cap T) + f(S \cup T) \leq f(S) + f(T)$ for all
$S,T \subseteq [n]$. We say that this function is \emph{monotone} if $f(S) \leq
f(T)$ for $S \subseteq T \subseteq [n]$.

 Applications of polymatroidal environments are ubiquous : multi-unit auctions
\cite{dobzinski12}, the matching markets \cite{fiat_clinching}, sponsored
search \cite{goel12, henzinger11}, bandwidth markets, scheduling with
deadlines, network planning, video on demand \cite{Bikhchandani11}, among
others. We refer to cite{goel12} and \cite{Bikhchandani11} for a more extensive
discussion on these applications.

\subsection{Average Budget Constraints}

An important special case of constrained utility functions are \emph{average
budget} constraints, in which $\A_i = \{(x_i, \pi_i) \in \R^2_+; \pi_i \leq
\beta_i \cdot x_i \}$, where the parameter $\beta_i$ is called the average
budget. In this particular case, the utility function enforces a constraint
that the player must pay at most $\beta_i$ per unit.

At a first glance, this might seem equivalent to a player being quasi-linear with
value $\tilde{v}_i = \min\{v_i, \beta_i\}$. In order to see the difference,
consider two different agents: (1) agent one has value $v_1 = 2$ and
average budget $\beta_1 = 1$ and (2) agent two is quasi-linear with value $v_2 =
\min \{v_1, \beta_1\} = 1$. Now, consider the outcome with $x_i = 1$ and $\pi_i
= 0$. At this point, consider offering an additional item for each player at
the price of $2$. The first player would gladly take it, since it is below his
value and doesn't violate the average budget constraints, but the second player
wouldn't.

Despite of that, both settings are not completely dissimilar. Consider the
problem of designing an incentive compatible, individually rational
and Pareto-efficient auction to sell one single good to players with average
budgets. It is simple to see that running the Vickrey auction on $\tilde{v}_i =
\min\{v_i, \beta_i\}$ does the trick. Or, more generally:

\begin{lemma}\label{lemma:multi_unit_average}
 For multi-unit auctions, i.e., $P = \{x\in \R^n_+; \sum_i x_i \leq s\}$, the
VCG auction on $\tilde{v}_i = \min\{v_i, \beta_i\}$ is incentive
compatible, individually rational and Pareto-efficient auction.
\end{lemma}

The proof is trivial and is included in the appendix for completeness. The
strategy above, however, doesn't generalize beyond multi-unit auctions.
Running VCG on $\tilde{v}_i$ for more general
polymatroidal environments is still incentive compatible and
individually rational, but fails to be Pareto-efficient in general. In fact,
consider the following very simple example:

\begin{example}
Consider the environment $P = \{x \in \R^2_+; x_1 + x_2 \leq 3, x_1 \leq 2, x_2
\leq 2\}$. For readers familiar with sponsored search, this corresponds to the
sponsored search environment with click-through-rates $(2,1)$. Now, consider
two agents with average budget constraints and $v_1 = 10, \beta_1 = 1$ and $v_2
= 2, \beta_2 = 2$. Running VCG on $\tilde{v}$ we get allocation $x = (1,2)$ and
$\pi = (0,1)$. This outcome is clearly not Pareto-efficient, since for any $0
< z \leq \frac{1}{2}$, the outcome $x = (1+z,2-z)$, $\pi = (2z, 1-2z)$ is a
Pareto-improvement $(x,\pi)$.
\end{example}

\section{Warm-up: the multi-units environment}

As a warm-up we consider the problem of designing an incentive
compatible, individually rational and Pareto optimal auction for the multi-units
setting, i.e., $P = \{x \in \R^n_+; \sum_i x_i\leq 1\}$ when agents have
constrained quasi-linear utilities. This will allow us to
highlight the main features of our design in a combinatorially simple setting.
The auction we will describe is a discrete step ascending clock price auction,
based on the clinching framework. The
auction takes as input the value $v_i$ of each agent and their admissible sets
$\A_i$ (defined in terms of an ability-to-pay function $\alpha_i$) and produces
a final allocation $x_i$ and payments $\pi_i$ for each agent. We will denote
$\beta_i = \lim_{x \downarrow 0} \alpha_i(x)/x$.

The auction is initialized with zero allocation and zero payments for all the
agents $x_i = \pi_i = 0$. The price clock is represented by a vector $p \in
\R^n_+$ where $p_i$ represents the price faced by agent $i$. Prices are 
initialized to zero.

In round-robin fashion an agent $\hat{i}$ is chosen and his price
$p_{\hat{i}}$ is incremented by a fixed amount $\epsilon > 0$. At this point, we
compute the
\emph{demand} of each agent, which is the maximum amount of the good that this
agent would want to acquire at price $p_i$, i.e., $d_i = \text{argmax}_z u_i(x_i
+ z, \pi_i + p_i z)$. It can be computed as follows: $d_i = \max \{z_i; (x_i
+ z_i, \pi_i + p_i z_i) \in \A_i\}$ if $p_i < v_i$ and $d_i = 0$ otherwise.
Based on the demands of each agent, we calculate how much each agent
\emph{clinches} in each round, i.e., how much it is safe to give to each agent
while not making any allocation infeasible for other agents.

The clinched amount is calculated as follows: let $s = 1-\sum_i x_i$ be the
remnant supply. The total demand of all agents except $i$ is given by $\sum_{j
\neq i} d_j$. We define the difference $\delta_i = [s - \sum_{j \neq i} d_j]^+$
as the clinched amount, i.e., the portion of the remnant supply that is
under-demanded by agents $[n] \setminus \{i\}$. The auction proceeds by updating
the allocation and payments by giving to each agent his clinched amount at the
current price. We summarize the auction in Algorithm
\ref{multi-units-clinching-auction}.

\begin{algorithm}[h]
 \caption{Multi-Units Clinching Auction}
  \textbf{Input:} $P, v_i, \A_i$\\
  $p_i = 0$, $x_i = 0$, for all $i$ and $\hi=1$\\
  \textbf{do} \\
   \algspace $d_i = \max \{z_i; (x_i + z_i, \pi_i + p_i z_i) \in \A_i \}$ if $p_i < v_i$ and $d_i = 0$ otherwise.\\
  \algspace $\delta_i = [1-\sum_j x_j - \sum_{j \neq i} d_j]^+$, \\
   \algspace $x_i = x_i + \delta_i$, \quad $\pi_i = \pi_i + p_i
\cdot \delta_i$,\\
  \algspace $d_i = \max \{z_i; (x_i + z_i, \pi_i + p_i z_i) \in \A_i \}$ if $p_i < v_i$ and $d_i = 0$ otherwise.
\comment{\checkpoint{\cp}} \\
  \algspace $p_{\hi} = p_{\hi} + \epsilon$ ,\quad $\hi = \hi
+ 1 \mod n$ \\
  \textbf{while} $d \neq 0$
\label{multi-units-clinching-auction}
\end{algorithm}

The outcome of the auction corresponds to the final allocation and payments of
the ascending procedure. It follows from standard arguments about the clinching
framework \cite{Ausubel_multi,dobzinski12} that this auction is incentive compatible
and individually rational. It is individually rational since it never allocates
any amount at a rate larger then the value $v_i$, so in the end, $\pi_i \leq v_i
x_i$. It is incentive compatible, since the value only determines when
an agent drops his demand to zero. By misreporting his value, 
an agent can either drop out of the ascending procedure earlier (potentially 
missing items he could acquire at a price smaller then his valuation) 
or drop later (potentially acquiring items for price larger then his value). 
Therefore it is a dominant strategy for each agent to truthfully report his value.

We are left to argue that the auction is Pareto optimal. In order to do that,
first, we define the notion of \emph{dropping price} and then we give a
structural characterization of the outcome in terms of dropping prices. The
main result in this paper (Lemma \ref{lemma:structure_tight_sets}) is a
generalization of this characterization for generic polymatroidal environments.

In the rest of the section, we will use the following assumption that holds
wlog in the limit when $\epsilon$ goes to zero:

\begin{assumption}
All values of $v_i$ and $\beta_i = \lim_{x \downarrow 0} \alpha_i(x)/x$
are 
multiples of $\epsilon$, which is the price clock increment in the 
ascending auction.
\end{assumption}

\subsection{Dropping Prices}\label{subsec:dropping_prices}

\begin{defn}[Dropping price]
Given an execution of Algorithm \ref{multi-units-clinching-auction}, we define
the \emph{dropping price} for agent $i$ (we call $\phi_i$) as the first price
for which he had zero demand.
\end{defn}

The demand of an agent can drop from
positive to zero in the execution of Algorithm
\ref{multi-units-clinching-auction} for three different reasons:

\begin{enumerate}
 \item the first case is where the buyer clinched his entire demand, i.e.,
$\delta_i = d_i$. By the definition
of demand, the player ends up with an allocation such that $\pi_i =
\alpha(x_i)$, since just before clinching, his demand was: $d_i = \max \{z_i;
(\tilde{x}_i + z_i, \tilde{\pi}_i + p \cdot z_i) \in \A_i\}$, where
$(\tilde{x}_i, \tilde{\pi}_i)$ is his allocation and payment just before
clinching for the last time.

After this happens, for any
price $p' \geq p$, the demand is zero, since a positive demand would imply that
there is some $\kappa > 0$ such that: $(\tilde{x}_i + d_i + \kappa,
\tilde{\pi}_i + p \cdot d_i + p' \cdot \kappa) \in \A_i$. This would contradict
the maximality of $z_i$, since by the concavity of $\alpha_i$, we
would have: $(\tilde{x}_i + (d_i + \kappa),
\tilde{\pi}_i + p \cdot (d_i+\kappa)) \in \A_i$. This is depicted in the
first part of Figure \ref{fig:dropping}.
 \item the player didn't clinch his entire demand, but the price reached his
value, i.e., $p = v_i$.
 \item the player didn't clinch his entire demand, but $\pi_i = \beta_i
\cdot x_i$ and $p > \beta_i$. This is depicted in the
second part of Figure \ref{fig:dropping}.
\end{enumerate}

We observe that:

\begin{lemma}\label{lemma:simple_bound}
 The dropping price $\phi_i$ is at most $v_i$. Also, if the final
outcome of agent $i$ is $x_i = \pi_i = 0$, then $\phi_i =
\min\{\beta_i+\epsilon, v_i\}$.
\end{lemma}

\begin{proof}
The fact that $\phi_i \leq v_i$ comes from the fact that $v_i$, $p_i$
and $\beta_i$ are multiples of $\epsilon$ and that for $p_i \geq v_i$, the
demand of $i$ is zero. Also, if $x_i = \pi_i = 0$, there are two reasons for
the demand to become zero: either $p_i$ becomes larger then $\beta_i$ or $p_i$
reaches $v_i$.
\end{proof}

\begin{figure}
\centering
\includegraphics{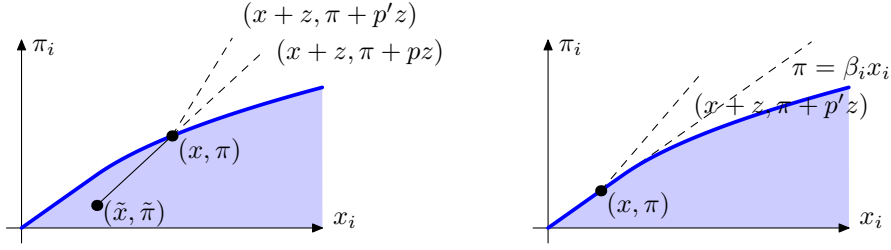}
\caption{Depiction of two reasons a player might drop his demand to zero. Left:
by clinching his entire demand. Right: by having $\pi_i = \beta_i x_i$ and $p >
\beta_i$. }
\label{fig:dropping}
\end{figure}

\subsection{Multi-units version of the Structure of Tight Sets Lemma}

Now we relate the dropping prices to the structure of the final outcome: we
show that if one sorts agents by dropping price, the agents with high dropping
price (which we will call $H$) will be allocated and charged the maximum
admissible amount fo the quantity they get. Agents with low dropping price
(which we will call $L$) will be unallocated and will be charged zero.

\begin{lemma}\label{lemma:primitive_sets_lemma}
Let $(x,\pi)$ be the outcome of the clinching auction for the multi-units
setting, then one can partition the set of agents $[n] = L \cup \{k\} \cup H$
such that:
\begin{itemize}
 \item for $i \in L$, $x_i = \pi_i = 0$ and $\phi_i \leq \phi_k$.
 \item for $i \in H$, $v_i > \phi_i \in \{ \phi_k, \phi_k - \epsilon\}$ and
$\pi_i = \alpha_i(x_i)$.
  \item $k$ drops without clinching his demand, therefore either (i) $\phi_{k} =
v_k$ or (ii) $\pi_k = \beta_k x_k$ and $\phi_k = \beta_k + \epsilon$.
\end{itemize}
Moreover, the clinching auction allocates all the goods, i.e., $\sum_i x_i = 1$.
\end{lemma}

{
The lemma above is a special case of Lemma \ref{lemma:structure_tight_sets}, so
we defer a formal proof until that point.  The proof of the special case is
implicit in Bhattacharya et al \cite{Bhattacharya10} and Goel et al
\cite{goel13}. The main idea behind it is to
show that all agents that acquire a positive amount drop their demand
at essentially the same price: once a player that already acquired a positive
amount drops his demand to zero, all other players clinch their entire demand.
}

\subsection{Pareto optimality}

\begin{theorem}
 The clinching auction with constrained quasi-linear utilities for the
multi-unit setting is Pareto optimal.
\end{theorem}

\begin{proof}
Let $(x,\pi)$ be the outcome of Algorithm \ref{multi-units-clinching-auction}
for valuations $v_i$ and admissible sets $\A_i$ defined by $\alpha_i$. Assume
that an alternative outcome $(x',\pi')$ is a Pareto improvement, i.e.,
$u_i(x'_i, \pi'_i)  \geq u_i(x_i, \pi_i)$, $\sum_i \pi'_i \geq \sum_i \pi_i$
and $\sum_i v_i x'_i > \sum_i v_i x_i$.

Let $L$, $H$ and $k$ be as in Lemma \ref{lemma:primitive_sets_lemma}. First we
show that $\pi_i - \pi'_i \geq \phi_k (x_i - x'_i)$ for all $i$ and for the
case where $x_i < x'_i$, the inequality holds strictly, i.e., $\pi_i - \pi'_i
> \phi_k (x_i - x'_i)$. Consider the following cases:
\begin{itemize}
 \item $i \in L$, then  $\pi'_i \leq
\min\{\beta_i, v_i\} \cdot x'_i \leq \phi_i x'_i \leq \phi_k x'_i$, where the
first inequality follows from the fact that $u_i(x'_i, \pi'_i) \geq 0$, the
second follows from Lemma \ref{lemma:simple_bound} and the third from Lemma
\ref{lemma:primitive_sets_lemma}. Noting that for $i \in L$, $x_i = \pi_i = 0$,
we get: $\pi_i - \pi'_i \geq \phi_k (x_i - x'_i)$
 \item $i \in H$, $x_i \geq x'_i$: since $v_i > \phi_i \geq \phi_k - \epsilon$
and all values and prices are multiples of $\epsilon$, then: $v_i \geq \phi_k$.
Since $u_i(x'_i, \pi'_i) \geq u_i(x_i, \pi_i)$, then $\pi_i - \pi'_i \geq
v_i \cdot (x_i - x'_i) \geq \phi_k \cdot (x_i - x'_i)$.
 \item $i \in H$, $x_i < x'_i$: player $i$ clinched his entire demand at
price $\phi_i$. By the definition of demand for any $\kappa > 0$, $(x_i +
\kappa, \pi_i + \phi_i \cdot \kappa) \notin \A_i$. In particular, for $\kappa =
x'_i - x_i$. Since $(x'_i, \pi'_i) \in \A_i$, it must be the case that: $\pi'_i
< \pi_i + \phi_i \cdot (x'_i - x_i)$. Now using the fact that $\phi_i \leq
\phi_k$ and re-arranging the inequality we get $\pi_i - \pi'_i >
\phi_{k}(x_i - x'_i)$.
\item $i = k$: then either (i) $\phi_k = v_k$, in which case we use the fact
that $(x',\pi')$ is a Pareto improvement to get that: $\pi_k - \pi'_k \geq v_k
(x_k - x'_k) = \phi_k (x_k - x'_k)$; or (ii) $\pi_k = \beta_k x_k$ and $\phi_k =
\beta_k + \epsilon$. If $x_k \geq x'_k$, we use the same argument as in the
second item. If $x_k < x'_k$, we use that $(x'_k, \pi'_k)$ is admissible
therefore $\pi'_k \leq \beta_k x'_k$ so: $\pi_k - \pi'_k \geq \beta_k (x_k -
x'_k) > \phi_k (x_k - x'_k)$.
\end{itemize}

Summing for all $i$ we obtain $\sum_i \pi_i - \sum_i \pi'_i \geq \phi_k \cdot
(\sum_i x_i - \sum_i x'_i) = \phi_k(1-\sum_i x'_i) \geq 0$. Since $\sum_i \pi_i
\leq \sum_i \pi'_i$, the revenue in both cases must be equal, therefore all
inequalities must hold with equalities. In particular, it must be that  $x_i
\geq x'_i$ for all $i$ since for $x_i < x'_i$, the inequality $\pi_i - \pi'_i
\geq \phi_k (x_i - x'_i)$ holds strictly. This in particular implies that
$\sum_i v_i x_i \geq\sum_i v_i x'_i$, contradicting the fact that $(x', \pi')$
is a Pareto improvement.
\end{proof}

\section{Polymatroidal environments}

Now we extend the result in the previous section to
general polymatroidal environments. We do so by changing the
way demands are calculated in the polyhedral clinching auction of Goel et al
\cite{goel12}. As usual, incentive compatibility and individual rationality
follow as usual from properties of the clinching framework.
The main challenge in extending the result in the previous section to
general polymatroidal environments is extending Lemma
\ref{lemma:primitive_sets_lemma} to combinatorial settings.

We begin by describing the polyhedral clinching auction for the case of
constrained quasi-linear utilities. The auction
takes as input the feasible set $P \subseteq \R^n_+$, agent values
$v_i$ and
valid admissible sets $\A_i$ and computes an allocation $x \in P$ and a
payment vector $\pi$ such that $(x_i, \pi_i) \in \A_i$ for all $i$.

The auction, described in Algorithm \ref{polyhedral-clinching-auction}, is a
version of Algorithm \ref{multi-units-clinching-auction} that redefines the
clinching step to take into account the environment $P$.

\begin{algorithm}[h]
 \caption{Polyhedral Clinching Auction}
  \textbf{Input:} $P, v_i, \A_i$\\
  $p_i = 0$, $x_i = 0$, for all $i$ and $\hi=1$ \\
  \textbf{do} \\
  \algspace $d_i = \max \{z_i; (x_i + z_i, \pi_i + p_i z_i) \in \A_i \}$ if $p_i
< v_i$ and $d_i = 0$ otherwise.\\
  \algspace $\delta = \text{\textbf{clinch}}(P,x,d)$, \\
   \algspace $x_i = x_i + \delta_i$, \quad $\pi_i = \pi_i + p_i
\cdot \delta_i$,\\
  \algspace $d_i = \max \{z_i; (x_i + z_i, \pi_i + p_i z_i) \in \A_i \}$ if $p_i
< v_i$ and $d_i = 0$ otherwise. \checkpoint{\cp} \\
  \algspace $p_{\hi} = p_{\hi} + \epsilon$ ,\quad $\hi = \hi
+ 1 \mod n$ \\
  \textbf{while} $d \neq 0$
\label{polyhedral-clinching-auction}
\end{algorithm}

\begin{defn}[Clinching]
 Given an allocation $x$ and demands $d$, the \emph{remnant supply} polytope is
defined as $P_{x,d} = \{y \in \R^n_+; x + y \in P; y \leq d \}$. Given an
amount $z_i$ for player $i$ we define the polytope on $P^i_{x,d}(z_i)$ of the
possible allocations for $[n]\setminus i$ if we allocate extra $z_i$ units to
player $i$. Formally: $P^i_{x,d}(z_i) = \{z_{-i} \in \R^{[n] \setminus i}_+;
(z_i, z_{-i}) \in P_{x,d}\}$. Now, the clinching amount $\delta_i$ is defined
as the maximum allocation to player $i$ that doesn't make any allocation for
other players infeasible: $\delta_i = \max \{z_i; P^i_{x,d}(z_i) =
P^i_{x,d}(0)\}$.
\end{defn}

It follows from standard arguments on clinching auctions that the auction is
incentive-compatible, individually-rational and produces admissable outcomes.
See for example Lemmas 3.3, 3.4 and 3.5 in \cite{goel12}. 

\begin{theorem}[\cite{goel12}]
 The clinching procedure is well defined (i.e. it stops after finite time and
$x \in P$). The auction produced is truthful, individually rational
and produces acceptable outcomes, i.e., $(x_i, \pi_i) \in \A_i$.
\end{theorem}

The rest of the paper is dedicated to prove that the auction produces
Pareto-efficient outcomes. The proof of Pareto-efficiency is based on a 
\emph{structural lemma} that relates the tight sets (i.e., sets of agents where
$\sum_{i \in S} x_i = f(S)$) to the dropping price as defined in Section
\ref{subsec:dropping_prices}. We note that the definitions and observations in
that section are valid for any environment.\\

Before we get to those, we introduce some new notation. For a vector $x \in
\R^n$ and a subset $S \subseteq [n]$ we denote $x(S) := \sum_{i \in S} x_i$.
Also, for the remainder of the paper, we focus on a polymatroidal environment
$P$ defined by a submodular function $f$, i.e., $P = \{x \in \R^n_+; x(S) \leq
f(S), \forall S \subseteq [n] \}$.

We will also keep the same notation used in the previous section: $\phi_i$ for
dropping prices, $\beta_i$ for $\lim_{x \downarrow 0} \alpha_i(x)/x$. We will
also assume, as before, that $v_i$ and $\beta_i$ are multiples of $\epsilon$.

\subsection{Pareto efficiency via Structure of tight sets}

Now we are ready to state the central pieces used to prove Pareto efficiency:

\begin{lemma}[Structure of tight sets]\label{lemma:structure_tight_sets}
 In any given iteration of Algorithm \ref{polyhedral-clinching-auction}, after
clinching is performed and demands are re-calculated (point \cp in the
description), if $S$ is the set of players with positive demand, then $S$
is tight in the final outcome, i.e, $x(S) = f(S)$ where $(x,\pi)$ is the final
outcome.
\end{lemma}

\begin{lemma}\label{lemma:dropping_prices}
 If a player clinches his entire demand in a certain iteration of the auction,
then there is a player that in the same iteration drops without clinching his
entire demand. 
\end{lemma}

We defer proving those lemmas until the end of Section~\ref{sec:saturation}.
Before, we discuss how to use this Lemma to prove Pareto-efficiency. Since it is
clear that the set of players with positive demand is shrinking, it gives us a
natural nested family of tight sets. Moreover, it is a family where we can
bound the prices for which they acquire goods during the process:

\begin{cor}\label{cor:structure}
Given an execution of \ref{polyhedral-clinching-auction}, let $i_1, \hdots,
i_k$ be the agents who drop their demand to zero without clinching their entire
demand sorted in reverse order in which those events (demand dropping to
zero) happens. Also, let $\phi_{i_1} \geq \hdots \geq \phi_{i_k}$ be the
dropping prices for each agent. Then by taking $S_j$ to be the set of agents
with positive demand just before player $i_j$ dropped his demand to zero, we
have a nested family of tight subsets:
$$\emptyset = S_0 \subset S_1 \subset S_2 \subset \hdots \subset S_k = [n],
\qquad x(S_j) = f(S_j), \forall j$$
with the property that: $i_j \in T_j := S_j \setminus S_{j-1}$ and all
players $i \in T_j \setminus \{ i_j \}$ clinched their entire demand and
$\phi_i \in \{ \phi_{i_j} - \epsilon, \phi_{i_j} \}$.
\end{cor}

\begin{proof}
 The proof follows directly from Lemma \ref{lemma:dropping_prices}. The fact
that for $i \in T_j$, $\phi_i \in \{\phi_{i_j} - \epsilon, \phi_{i_j}\}$ follows
from Lemma \ref{lemma:structure_tight_sets} since in the iteration player $i_j$
drops, the prices of all the agents are either $\phi_{i_j}$ or $\phi_{i_j} -
\epsilon$.
\end{proof}

The reader familiar with \cite{goel12} will note the similarity of the nested
family in Corollary \ref{cor:structure} and the nested family in the proof
of Lemma 3.8 in \cite{goel12}. From the proof in \cite{goel12} it should be
clear that getting such a tight family is the main ingredient in proving
Pareto-efficiency. The difficulty here is that we need a family which is somehow
tied to prices, which wasn't necessary in \cite{goel12}. There, one could simply
use values $v_i$ are proxies prices $\phi_i$, since the admissible sets were
very simple.

\begin{theorem}
 The outcome of the Polyhedral Clinching Auction (Algorithm
\ref{polyhedral-clinching-auction}) is Pareto-efficient.
\end{theorem}

\begin{proof}
 Let $(x,\pi)$ be the outcome of the clinching auction and assume that there is
an alternative outcome $(x', \pi')$ such that $v_i \cdot x'_i - \pi'_i \geq v_i
\cdot x_i - \pi_i, \forall i$, $\sum_i \pi'_i \geq \sum_i \pi_i$ and at least
one of those inequalities is strict, which means that the sum of those
inequalities is a strict inequality: $\sum_i v_i \cdot x'_i > \sum_i v_i \cdot
x_i$. Also, let $x'$ is a feasible point of $P$, $(x'_i, \pi'_i) \in \A_i,
\forall i$.

Assuming the structure in Corollary \ref{cor:structure}, first we show that for
$i \in T_j$, $\pi_i - \pi'_i \geq \phi_{i_j}(x_i - x'_i)$. In order to show
that, we consider three cases:
\begin{itemize}
 \item $x_i \geq x'_i$, $i \neq i_j$ : since $i$ clinched his entire demand,
then $\phi_i \leq v_i - \epsilon$ and since $\phi_i \in \{ \phi_{i_j} -
\epsilon, \phi_{i_j}
\}$, then: $\phi_{i_j} \leq v_i$. Then $\pi_i - \pi'_i \geq v_i \cdot (x_i -
x'_i) \geq \phi_{i_j} \cdot (x_i - x'_i)$.
 \item $x_i < x'_i$, $i \neq i_j$ : player $i$ clinched his entire demand at
price $\phi_i$. By the definition of demand for any $\kappa > 0$, $(x_i +
\kappa, \pi_i + \phi_i \cdot \kappa) \notin \A_i$. In particular, for $\kappa =
x'_i - x_i$. Since $(x'_i, \pi'_i) \in \A_i$, it must be the case that: $\pi'_i
< \pi_i + \phi_i \cdot (x'_i - x_i)$. Now using that $\phi_i \leq
\phi_{i_j}$ and re-arranging the inequality we get $\pi_i - \pi'_i >
\phi_{i_j}(x_i - x'_i)$.
 \item $i = i_j$ : in this case: (i) either $\phi_{i_j} = v_{i_j}$, in which
case $\pi_{i_j} - \pi'_{i_j} \geq v_{i_j} \cdot (x_{i_j} - x'_{i_j}) =
\phi_{i_j} \cdot (x_{i_j} - x'_{i_j})$, (ii) or $\phi_{i_j} > \beta_{i_j}$ and
$\pi_{i_j} = \beta_{i_j} \cdot x_{i_j}$. Since $\pi'_{i_j} \leq \beta_{i_j}
\cdot x_{i_j}$, we have that: $\pi_{i_j} - \pi'_{i_j} \geq \beta_{i_j} \cdot
(x_{i_j} - x'_{i_j}) > \phi_{i_j} \cdot (x_{i_j} - x'_{i_j})$ if $x_{i_j} 
\leq x'_{i_j}$. In the case where $x_{i_j} \geq x'_{i_j}$, we can simply use
that $\phi_{i_j} \leq v_{i_j}$ and then: $\pi_{i_j} - \pi'_{i_j} \geq v_{i_j}
\cdot (x_{i_j} - x'_{i_j}) \geq \phi_{i_j} \cdot (x_{i_j} - x'_{i_j})$.
\end{itemize}

Now, summing this inequality for all $i \in T_j$, we get:
$$\sum_{i \in T_j} \pi_i - \pi'_i \geq \phi_{i_j} \cdot \sum_{i \in T_j} x_i -
x'_i = \phi_{i_j} \cdot (x(T_j) - x'(T_j)) \geq 0$$
since $x(T_j) = f(T_j) \geq x'(T_j)$. Therefore:
$$\sum_{i \in S_j} \pi_i - \pi'_i = \sum_{k \leq j} \sum_{i \in T_k} \pi_i -
\pi'_i \geq \phi_{i_j} \cdot (x(S_j) - x'(S_{j})) +  \sum_{k < j} (\phi_{i_k} -
\phi_{i_{k+1}})\cdot (x(S_k) - x'(S_{k})) \geq 0$$
For $j = k$, $S_k = [n]$, so: $0 \geq \sum_{i \in [n]} \pi_i - \pi'_i \geq 0$,
where the first inequality comes from Pareto-efficiency and the second
inequality comes from the line above. This means that all inequalities along
the way used to prove the inequality above should be tight. This means in
particular that for all $i \in T_j$, either $x_i = x'_i$ or $v_i = \phi_{i_j}$,
since $x_i \neq x'_i$ and $v_i > \phi_{i_j}$ would imply a strict inequality in
some of the cases above. Therefore:
$$\sum_{i \in T_j} v_i \cdot (x_i - x'_i) = \sum_{i \in T_j, v_i = \phi_{i_j}}
\phi_{i_j} \cdot (x_i - x'_i) = \sum_{i \in T_j}
\phi_{i_j} \cdot (x_i - x'_i) \geq \phi_{i_j} \cdot (f(T_j) - x'(T_j)) \geq 0$$
Summing for all $j$, we get that $\sum_i v_i \cdot x_i \geq \sum_i v_i
\cdot x'_i$ contradicting the assumption that  $\sum_i v_i \cdot x'_i >
\sum_i v_i \cdot x_i$.
\end{proof}

\noindent \textbf{Relation to the proof for hard budget constraints:} In the
proof of Lemma 3.8 in \cite{goel12} for $\A_i = \{(x_i, \pi_i) \in \R^2_+; \pi_i
\leq B_i\}$, the case $x_i < x'_i, i \neq i_j$ is simpler, since one can use
that
$\pi_i = B_i \geq \pi'_i$ together with $v_i x_i < v_i x'_i$ to prove that:
$\pi_i - \pi'_i \geq 0 > v_i (x_i - x'_i)$. The entire proof can be done
therefore using values $v_{i}$ instead of dropping prices $\phi_i$. Since there
is not a global upper bound on payment, this approach
does not work for a general admissible set.

\subsection{Basic Facts on Polymatroids and Clinching}

In the previous subsection we showed that the Pareto-optimality follows from
Lemmas \ref{lemma:dropping_prices} and \ref{lemma:structure_tight_sets}.
Proving those statements is the most technically challenging part of the paper.
Before we do it, we would like to review some elementary facts about
polymatroids. See \cite{schrijver-book} for example, for an extensive exposition
on polymatroids.

\begin{lemma}[Uncrossing]\label{lemma:uncrossing}
 If $P$ is the polymatroid defined by $f:2^{[n]} \rightarrow \R_+$, and $x \in
P$ such that for subsets $S, T \subseteq [n]$, $x(S) = f(S)$ and $x(T) = f(T)$
(we say those sets are tight), then $x(S \cap T) = f(S \cap T)$ and $x(S \cup
T) = f(S \cup T)$.
\end{lemma}

\begin{lemma}[Polymatroid $\cap$ Box]\label{lemma:polymatroid_box}
 If $P$ is the polymatroid defined by $f:2^{[n]} \rightarrow \R_+$, then
$P_{x,d} = \{y \in \R^n_+; x+y \in P; y \leq d \}$ is the polymatroid defined by
the (possibly non-monotone) submodular function $\tilde{f}(S) = \min_{T
\subseteq S} [ f(T) - x(T) + d(S \setminus T) ]$.
\end{lemma}

Now, we also review a basic fact about clinching, which was proved in
\cite{goel12} :

\begin{lemma}[Constructive Clinching]\label{lemma:constructive_clinchign_1}
 Given current allocation $x$ and payments $\pi$ and a polymatroid $P$, the
clinched amount $\delta_i$ can be calculated as $\delta_i = [ \max_{y \in
P_{x,d}} \one^t y ] - [ \max_{y \in P_{x,d}} \one^t_{-i} y_{-i} ]$. An
alternative description is: $\delta_i = [\tilde{f}([n]) - \tilde{f}([n]
\setminus i)]^+$
for the $\tilde{f}$ function defined in the previous lemma.
\end{lemma}

We note that the $\tilde{f}$ function define in Lemma
\ref{lemma:polymatroid_box}
might not be monotone. The non-monotonicity has to do with the $-x(T)$ term in
the definition. It simple to see that if $x = 0$, then $\tilde{f}(\cdot)$ is
monotone, since for $S \subseteq S'$ : $\tilde{f}(S') = f(T) + d(S' \setminus
T) \geq f(T \cap S) + d(S \setminus T) \geq \tilde{f}(S)$, where $T$ is the
subset of $S'$ minimizing $f(T) + d(S' \setminus T)$. The following lemma will
allow us to define clinching in terms of a monotone submodular function:

\begin{lemma}\label{lemma:constructive_clinchign_2}
 Given a polymatroid $P$ defined by $f$, $x \in P$ and a demand vector $d \in
\R^n_+$, then: $$\max_{y \in P_{x,d}} \one^t y = \max_{y \in P_{0,x+d}} \one^t y
- \one^t x$$
\end{lemma}

\begin{proof}
 The problem $\max_{y \in P_{x,d}} \one^t y$ can be written as $\max \one^t y
\text{ s.t. } (x+y)(S) \leq f(S); \forall S; 0 \leq y \leq d$. Once we relax $0
\leq y$ to $-x \leq y$ and substitute $z = x+y$ we get the problem: $\max
\one^t (z-x) \text{ s.t. } z(S) \leq f(S); \forall S; 0 \leq z \leq x+d$. So,
it is simple to see that $\max_{y \in P_{0,x+d}} \one^t y - \one^t x$ is a
relaxation of the first problem. Now, among the \emph{optimal} solutions $z$ to
the first problem, take one minimizing $\Phi(z) = \sum_i (x_i - z_i)^+$, i.e.,
take the optimal solution that minimally violates $x \leq z$. We wish
to prove that $\Phi(z) = 0$.

Assume by contradiction that $\Phi(z) > 0$. If this is true, then $S^- = \{i;
z_i < x_i\} \neq \emptyset$. Since $\one^t z \geq \one^t x$ (after all $z=x$ is
feasible), then $S^+ = \{i; z_i > x_i \} \neq \emptyset$ as well. Now, notice
that for all $i \in S^-$, $j \in S^+$ and $\delta > 0$, the solution $(z_i +
\delta, z_j - \delta, z_{-i,j})$ can't be feasible, otherwise we would violate
the minimality of $\Phi$. Therefore, there must be a tight set $T_{ij}$ between
$i$ and $j$, i.e., $z(T_{ij}) = f(T_{ij})$, $i \in T_{ij}$, $j \notin T_{ij}$.
Using uncrossing (Lemma \ref{lemma:uncrossing}), the set $T = \cup_{i \in S^-}
\cap_{j \in S^+} T_{ij}$ we get a set that is tight, i.e., $z(T) = f(T)$, $S^-
\subseteq T$ and $T \cap S^+ = \emptyset$. So, $f(T) = z(T) < x(T)$, where the
second inequality comes from the fact that $S^- \subseteq T \subseteq [n]
\setminus S^+$. This contradicts the fact that $x(T) \leq f(T)$ because $x \in
P$.
\end{proof}

The previous lemma motivates the following notation for submodular functions
capped by a vector: given a vector $\psi \in \R^n_+$, define $f_\psi(S) =
\min_{T \subseteq S} f(T) + \psi(S \setminus T)$, which is the submodular
function defining $P_{0,\psi}$. Using this new notation together with the
previous lemmas, we get:

\begin{cor}
 The clinched amount can be calculated as: $\delta_i = f_{x+d}([n]) -
f_{x+(0,d_{-i})}([n])$.
\end{cor}

\begin{proof}
 Notice that $\max_{y \in P_{x,d}} \one^t_{-i} y_{-i} = \max_{y \in
P_{x,(0,d_{-i})}} \one^t y = f_{x+(0,d_{-i})}([n]) - \one^t x$
\end{proof}

\subsection{Saturation and a proof of the Structure of Tight Sets Lemma}
\label{sec:saturation}

This sets the stage to the definition of \emph{saturation}, which will be
fundamental concept in the following proofs. First we give an intuitive notion
of saturation and then we give a more practical equivalent definition using the
$f_\psi$ notation.

\begin{defn}[Saturation]
 Given $x,d$ in a certain point of the execution of the Polyhedral Clinching
Auction (Algorithm \ref{polyhedral-clinching-auction}) we say that agent $i$ is
saturated if there is an optimal solution to $\max_{z \in P_{0,x+d}} \one^t z$
with $z_i < x_i + d_i$. We say that $i$ is unsaturated if all optimal solutions
are such that $z_i = x_i + d_i$.

Also, for a fixed agent $k$, we say that agent $i$ is $k$-saturated if, for
$\psi$ such that $\psi_{-k} = x_{-k} + d_{-k}$ and $\psi_k = x_k$, there
are optimal solutions to $\max_{z \in P_{0,\psi}} \one^t z$ with $z_i < 
\psi_i$. We say that $i$ is $k$-unsaturated if all optimal solution are such
that $z_i = \psi_i$. 
\end{defn}

{\noindent \textbf{Intuition of the concept of saturation and a
connection to previous work:} For the special case of multi-unit auctions $P =
\{x; \sum_i x_i \leq 1\}$ studied in \cite{dobzinski12, Bhattacharya10, goel13},
an important concept is that of the clinching set -- which is the set of players
who clinch some amount of the good as the price increases. For generic
polymatroidal setting, the concept of saturation emulates the concept of the
clinching set in the following sense:
player $i$ will be able to clinch as the demand of $k$ drops iff $k$ is
$i$-unsaturated. For multi-unit auctions, the set of $k$-unsaturated elements
are either $\emptyset$ or $[n]$. So, one can represent this structure by
defining the \emph{clinching set} as the set of players $k$ for which all $[n]$
are $k$-unsaturated.  \comment{This connection is formalized in the proof of
Lemma \ref{lemma:structure_tight_sets}.}\\}

The next lemma gives a more practical definition of saturation:

\begin{lemma}
 Let $\psi$ be such that $\psi_{-k} = x_{-k} + d_{-k}$ and $\psi_k = x_k$. Then
agent $i$ is $k$-unsaturated iff $\psi_i = f_\psi([n]) - f_\psi([n]\setminus
i)$.
\end{lemma}

\begin{proof}
 The optimal solution of $\max_{z \in P_{0,\psi}} \one^t z$ that minimizes
$z_i$ is the one that maximizes $\one^t_{-i} z_{-i}$. Since the feasible set is
a polymatroid, one can simply greedly increase each coordinate as much as one
can, leaving $i$ as the last one. Therefore we get $z_i =  f_\psi([n]) -
f_\psi([n]\setminus i)$. Now, if $z_i = \psi_i$, then $i$ is $k$-unsaturated,
if $z_i < \psi_i$, then $i$ is $k$-saturated.
\end{proof}

Now, we state and prove two useful lemmas on dealing with submodular functions
capped by a vector:

\begin{lemma}\label{lemma:auxiliary}
 Given $\psi$ and $\psi' = (\psi'_i, \psi_{-i})$ with $\psi'_i < \psi_i$ and $i
\in S$, the following identity holds: $f_{\psi'}(S) = \min \{ f_\psi(S),
f_\psi(S \setminus i) + \psi'_i \}$.
\end{lemma}

\begin{proof}
 By the definition of $f_{\psi'}$, there is a set $T \subseteq S$ such that
$f_{\psi'}(S) = f(T) + \psi'(S \setminus T)$. If $i \in T$, then $f_{\psi'}(S)
= f_{\psi}(S)$. If not, then: $f_{\psi'}(S) = \psi'_i + f(T) + \psi'(S
\setminus T,i) = \psi'_i + f(T) + \psi(S \setminus T,i) = \psi'_i + f_{\psi}(S
\setminus i)$ by the fact that $T$ minimizes $f(T) + \psi(S \setminus T,i)$ for
$T \subseteq S \setminus i$.
\end{proof}

\comment{
\begin{lemma}
 Given an allocation $x$ and demands $d$ in a certain point of the algorithm,
let $U$ be the set of $k$-unsaturated elements and $S$ be the
set of $k$-saturated elements. If $\psi = (x_k, x_{-k}+d_{-k}) \in \R^n_+$,
then: $f_\psi ([n]) = f(S) + \psi(U)$.
\end{lemma}

\begin{proof}
 By the definition of unsaturated elements, $i \in U$ if  $\psi_i = f_\psi([n])
- f_\psi([n]\setminus i)$. By submodularity, for every $i \in S \subseteq [n]$,
$\psi_i \geq f_\psi(S)-  f_\psi(S\setminus i) \geq \psi_i$, so: $f_\psi(S)- 
f_\psi(S\setminus i) = \psi_i$. So: Let $i_1, \hdots, i_k$ be the elements in
$U$, so:
$$f_\psi([n]) = f_\psi(S) + \sum_{j=1}^k f_\psi(S \cup \{i_1,\hdots, i_j\}) -
f_\psi(S \cup \{1_1, \hdots, i_{j-1}\}) = f_\psi(S) + \sum_{j=1}^k \psi_{i_j} =
f_\psi(S) + \psi(U)$$
Now we need to show that $f_\psi(S) = f(S)$. In fact, by the definition of
$f_\psi$, $f_\psi(S) = f(T) + \psi(S \setminus T)$. If $S \setminus T \neq
\emptyset$, then take $i \in S \setminus T$, so: $ \psi_i \geq f_\psi([n]) -
f_\psi([n] \setminus i)  \geq f_\psi(S) - f_\psi(S\setminus i) =  f(T) + \psi(S
\setminus T) - f_\psi(S\setminus
i) \geq f(T) + \psi(S \setminus T) - [ f(T) + \psi(S \setminus T,i) ] = \psi_i
$, so $i$ is $k$-unsaturated, contradicting the definition of $S$.
\end{proof}
}

\begin{lemma}\label{lemma:decomposition}
 Given an allocation $x$ and demands $d$ in a certain point of the algorithm,
let $U$ be the set of $k$-unsaturated elements and $S$ be the
set of $k$-saturated elements. If $\psi = (x_k, x_{-k}+d_{-k}) \in \R^n_+$ and
$S \subseteq X \subseteq [n]$, then: $f_\psi (X) = f(S) + \psi(X \cap
U)$.
\end{lemma}

\begin{proof}
 By the definition of unsaturated elements, $i \in U$ if  $\psi_i = f_\psi([n])
- f_\psi([n]\setminus i)$.  By submodularity, for every $i \in X \subseteq [n]$,
$\psi_i \geq f_\psi(X)-  f_\psi(X\setminus i) \geq \psi_i$, so: $f_\psi(X)- 
f_\psi(X\setminus i) = \psi_i$. So: Let $i_1, \hdots, i_k$ be the elements in
$X \cap U$, so:
$$\begin{aligned}f_\psi(X) & = f_\psi(S) + \sum_{j=1}^k f_\psi(S
\cup \{i_1,\hdots, i_j\}) -
f_\psi(S \cup \{1_1, \hdots, i_{j-1}\}) =\\ & = f_\psi(S) +
\sum_{j=1}^k \psi_{i_j} =
f_\psi(S) + \psi(X \cap U)\end{aligned}$$
Now we need to show that $f_\psi(S) = f(S)$. In fact, by the
definition of
$f_\psi$, $f_\psi(S) = f(T) + \psi(S \setminus T)$. If $S \setminus T \neq
\emptyset$, then take $i \in S \setminus T$, so: $$\psi_i \geq f_\psi(S) -
f_\psi(S\setminus i) =
 f(T) + \psi(S\setminus T) - f_\psi( S\setminus
i) \geq f(T) + \psi( S \setminus T) - [ f(T) + \psi( S \setminus
(T \cup i)) ] = \psi_i$$
Using the identity we just proved, we have:
$$\psi_i \geq f_\psi([n]) - f_\psi([n]\setminus i) \geq [ f_\psi(S) + \psi([n]
\setminus S) ] - [f_\psi(S \setminus i) + \psi([n] \setminus S)] = \psi_i $$
which leads to the conclusion that $i$ is $k$-unsaturated, contradicting the
definition of $S$. Therefore it must be the case that $S \setminus T =
\emptyset$, i.e., $f_\psi(S) = f(S)$.
\end{proof}

Now, we state a sequence of three invariants that are mantained throughout the
auction. The first two of them are proved in \cite{goel12}, while the third one
has to do with the concept of saturation introduced in this paper. The proofs
can be found in Appendix \ref{appendix:proofs}.

\begin{lemma}[Invariant I: Maximality of clinching]\label{lemma:maximality}
Performing the clinching step twice in a row without updating prices
will result in no amount clinched in the second step. Alternatively: in point
\cp of the execution of the algorithm: $f_{x+d}([n]) =
f_{x+(0,d_{-i})}([n])$ for all $i \in [n]$.
\end{lemma}

\begin{lemma}[Invariant II: All goods sold]\label{lemma:all_goods_sold}
 At any point of the execution, it is always possible to sell all the goods.
Alternatively: $f_{x+d}([n]) = f([n])$.
\end{lemma}

\begin{lemma}[Invariant III: Self-unsaturation]\label{lemma:self-saturation}
 Through the execution of the clinching auction, player $i$ is $i$-unsaturated
for every $i$. Alternatively: $f_{x+(0,d_{-i})}(S) = f_{x+(0,d_{-i})}(S
\setminus i) + x_i$ for all $i \in [n]$ and $S \ni i$.
\end{lemma}

Now, we are ready to prove Lemma \ref{lemma:dropping_prices}:

\begin{proofof}{Lemma \ref{lemma:dropping_prices}}
Let $x,d$ be the allocation and demands of players at point \cp of the
execution of Algorithm \ref{polyhedral-clinching-auction}. If before the time
the algorithm is in point \cp some player drops his demand to zero, we
will show that either this player is $\hi$ (the player for which the price
increased) or $d_\hi$ became zero.

If $d_\hi$ was zero already, nothing changes when the price $p_\hi$ increases,
so no clinching happens. On the other hand, if after $p_\hi$ increases and
demands are updated, the new demand of $\hi$ (let's call it $d'_\hi$) is still
positive, we will see that $\delta_i < d_i$ for all players with $i \neq 
\hi$ and $d_i > 0$.

Let's call $d' = (d'_\hi, d_{-\hi})$. So, $\delta_i = f_{x+d'}([n]) -
f_{x+(0,d'_{-i})}([n])$. First, notice that: $f_{x+d'}([n]) \geq
f_{x+(0,d_{-\hi})}([n]) = f_{x+d}([n]) \geq f_{x+d'}([n]) $, where the
inequalities come from $d' \geq (0,d_{-\hi})$ and $d \geq d'$. The equality
comes from  Lemma \ref{lemma:maximality}. Therefore: $f_{x+d'}([n]) =
f_{x+d}([n])$.

Also, by Lemma \ref{lemma:auxiliary}, $f_{x+(0,d'_{-i})}([n]) = \min \{
f_{x+(0,d_{-i})}([n]), f_{x+(0,d_{-i})}([n] \setminus \hi) + x_\hi + d'_\hi \}$
If $f_{x+(0,d'_{-i})}([n]) = f_{x+(0,d'_{-i})}([n])$, then: 
$\delta_i = f_{x+d}([n]) - f_{x+(0,d_{-i})}([n])  = 0 < d_i$ by Lemma
\ref{lemma:maximality}. If $f_{x+(0,d'_{-i})}([n]) = f_{x+(0,d_{-i})}([n]
\setminus \hi) + x_\hi + d'_\hi $, then:
$$\begin{aligned}
   \delta_i & = f_{x+d}([n]) - [f_{x+(0,d_{-i})}([n]
\setminus \hi) + x_\hi + d'_\hi] 
\stackrel{\text{\ref{lemma:self-saturation}}}{=} f_{x+d}([n]) -
[f_{x+(0,d_{-i})}([n]
\setminus \hi,i) + x_i + x_\hi + d'_\hi]
\stackrel{\text{\ref{lemma:maximality}}}{=} \\
& = f_{x+(0,d_{-\hi})}([n]) -
[f_{x+(0,d_{-\hi})}([n]
\setminus \hi,i) + x_i + x_\hi + d'_\hi] = \\ & = f_{x+(0,d_{-\hi})}([n]) -
f_{x+(0,d_{-\hi})}([n] \setminus \hi) + f_{x+(0,d_{-\hi})}([n] \setminus \hi) -
f_{x+(0,d_{-\hi})}([n] \setminus \hi,i) - [x_i + x_\hi + d'_\hi] 
\stackrel{\text{\ref{lemma:self-saturation}}}{\leq} \\ & \leq
x_\hi + x_i + d_i -[x_i + x_\hi + d'_\hi] = d_i - d'_\hi < d_i
  \end{aligned}
$$

\end{proofof}

\begin{proofof}{Lemma \ref{lemma:structure_tight_sets}}
Let $x,d$ be the allocation and demands during the execution of the
algorithm in point \cp. We will show if $T = \{i; d_i > 0\}$, then
$x([n] \setminus T) = f([n]) - f(T)$. If we show that, then we are done, since
for the players in $[n]\setminus T$, the value of $x$ coincides with the final
allocation, so, if $x^f$ is the final allocation, we know that $x^f([n]
\setminus T) = x([n] \setminus T)$. By Lemma \ref{lemma:all_goods_sold},
$x^f([n]) = f([n])$, so: $x^f(T) = f(T)$.

In the first iteration $T = [n]$, so the property $x([n]
\setminus T) = f([n]) - f(T)$ holds trivially. Now we only need to show
that this property is preserved from one iteration to another. Let $x,d$ be
the allocation and demands when the algorithm is at point \cp. Also, let
$\hi$ be the player whose price increases just after point \cp. If
$d_\hi$ doesn't drop to zero, no other player's demand drops to zero (Lemma
\ref{lemma:dropping_prices}), so $T$ remains unchanged and the property is
preserved. 

If on the other hand, $d_\hi$ becomes zero, we claim that the players whose
demand drops to zero are exactly the $\hi$-unsaturated players. Letting $d' =
(0, d_{-\hi})$ and then using Lemma \ref{lemma:auxiliary} we can see that:
$$\begin{aligned}\delta_i = f_{x+d'}([n])-f_{x+(0,d'_i)}([n]) & = f_{x+d'}([n])
- \min\{
f_{x+d'}([n]), f_{x+d'}([n]\setminus i) + x_i \} \\ & = [f_{x+d'}([n]) -
f_{x+d'}([n]\setminus i) - x_i]^+\end{aligned}$$
 Therefore $\delta_i = d_i$ iff $f_{x+d'}([n]) -
f_{x+d'}([n]\setminus i) = x_i + d_i$, i.e., $i$ is $\hi$-unsaturated.
Let $U$ and $S$
be respectively the sets of $\hi$-unsaturated in $T$ and $\hi$-saturated
players in $T$. Notice that the players in $U$ will be ones who will join the
set of players with zero demand in the next iteration and their final
allocation (let's call is $x^f$) will be $x^f_i = x_i + d_i$ for $i \in U
\setminus \hi$ and $x^f_\hi = x_\hi$. The set $S$ will be the players who will
have non-zero demand the next time we reach point \cp.

In order to use Lemma \ref{lemma:decomposition}, we note that the players
outside $T$ are $\hi$-unsaturated, since for $k \notin T$, $x+(0,d_{-\hi}) \leq
x+d =
x+(0, d_{-k})$, therefore: $f_{x+(0,d_{-\hi})}([n] \setminus k) \leq
f_{x+(0,d_{-k})}([n] \setminus k) = f_{x+(0,d_{-k})}([n]) - x_k$, since $k$ is
$k$-unsaturated. This means that: $f_{x+(0,d_{-\hi})}([n]) -
f_{x+(0,d_{-\hi})}([n] \setminus k)\geq f_{x+(0,d_{-k})}([n]) -
[f_{x+(0,d_{-k})}([n]) - x_k] = x_k$.
Therefore $k$ is $\hi$-unsaturated. Now, we can apply
Lemma \ref{lemma:decomposition}
that $f(S) + (x+d')(U) = f_{x+d'}(T) \leq f(T)$. Summing this with: $x([n]
\setminus T) = f([n]) -  f(T)$ we get: $x([n] \setminus T) + (x+d')(U) \leq
f([n]) - f(S)$. To see this holds with equality, let $x^f$ be the final
allocation of the algorithm and notice that: $f(S) \geq x^f(S) = f([n]) -
x^f([n] \setminus S) = f([n]) - [x([n] \setminus T) + (x+d')(U)]$. This shows
us that:
$$x^f([n]\setminus S) = x([n] \setminus T) + (x+d')(U) = f([n]) -
f(S)$$
which establishes that the invariant is preserved from one iteration to
another. 
\end{proofof}

\bibliographystyle{abbrv}
\bibliography{sigproc}

% You must have a proper ".bib" file
%  and remember to run:
% latex bibtex latex latex
% to resolve all references
%
% ACM needs 'a single self-contained file'!
%
%APPENDICES are optional
%\balancecolumns
\appendix

\section{Appendix: Missing Proofs}\label{appendix:proofs}

\begin{proofof}{Lemma \ref{lemma:multi_unit_average}}
 Let $\tilde{v}_1 \geq \tilde{v}_2 \geq \hdots \geq \tilde{v}_n$. The outcome
of VCG allocates all the units to player $1$ and nothing to the other players,
i.e., $x_1 = s$, $x_i = 0$ for $i \neq 1$. Also, $\pi_1 = \tilde{v}_2$ and
$\pi_i = 0$ for $i \neq 1$.

Let $(x',\pi')$ be a Pareto-improvement. Since the utility of $1$ improves, it
must be the case that: $v_1 x'_i - \pi'_1 \geq v_1 s - \tilde{v}_2 s$.
Therefore: $\pi'_1 \leq v_1 x'_1 - v_1 s + \tilde{v}_2 s$. Now, for each player
$i \neq 1$, $\pi_i \leq \beta_i x_i \leq \tilde{v}_2 x_i$. So: $\sum_i \pi'_i
\leq \tilde{v}_2 s + (\tilde{v}_2 - v_1) (s-x'_1)$. But since it is a
Pareto-improvement, $\sum_i \pi'_i \geq \sum_i \pi_i = \tilde{v}_2 s$, so $x'_1
= s$. Which implies that $x = x'$, $\pi = \pi'$, which is a contradiction.
\end{proofof}

\noindent \textbf{Note on the proof of Lemma \ref{lemma:multi_unit_average}}:
A alternative proof of Lemma \ref{lemma:multi_unit_average} is to note that for
multi-unit auctions and just average budget constraints, the polyhedral
clinching auction (Algorithm \ref{polyhedral-clinching-auction}) boils down to
VCG on $\min\{v_i,\beta_i\}$. Beyond multi-unit auctions, however, this is not
the case anymore.\\

\begin{proofof}{Lemma \ref{lemma:maximality}}
This is trivially true in the first time the algorithm reaches point
\cp, since at that point: $f([n]) = f_{x+d}([n]) = [f([n]) -
f([n]\setminus i)] + f([n]\setminus i) = x_i + f([n] \setminus i) =
f_{x+(0,d_{-i})}([n])$, after all, at that point $d_i = \infty$ for all $i$.
Now, we will show this invariant is preserved the next time point \cp is
reached again. 

Let $d'$ be the demand vector after the price of $\hi$ increased. Then clinched
amounts are: $\delta_i = f_{x+d'}([n]) - f_{x+(0,d'_{-i})}([n])$. The vector
$x$ is then updated to $x+\delta$ and $d'$ to $d'-\delta$. We want to show
that: $f_{(x+\delta)+(d'-\delta)}([n]) = f_{(x+\delta)+(0,d'_{-i} +
\delta_{-i})}([n])$. Which is :$f_{x+d'}([n]) =
f_{x +(\delta_i,d'_{-i})}([n])$. 

By Lemma \ref{lemma:auxiliary}, $f_{x+(0,d'_{-i})}([n]) = \min\{ f_{x+d'}([n]),
f_{x+d'}([n]\setminus i) + x_i \}$. If the minimum is the first term, then
$\delta_i = 0$ and we are done. If it is the second term, then: $\delta_i =
f_{x+d'}([n]) - f_{x+d'}([n]\setminus i) - x_i$, in which case:
$f_{x +(\delta_i,d'_{-i})}([n]) = \min \{f_{x +d'}([n]), f_{x
+d'}([n] \setminus i) + x_i + \delta_i \} = f_{x +d'}([n]) $.
\end{proofof}

\begin{proofof}{Lemma \ref{lemma:all_goods_sold}}
  See that this invariant is preserved through the execution of the algorithm.
It is trivially true in the beginning since $d_i = \infty$ for all $i$. Now, if
the invariant holds at point \cp, then by Lemma \ref{lemma:maximality}
$f([n]) = f_{x+d}([n]) = f_{x+(0,d_{-\hi})}([n])$. Now, when the price
increases, the demands are updated and clinching happens, the value of
$x+(0,d_{-\hi})$ doesn't change, since $x_\hi$ doesn't change and for $i \neq
\hi$, $x_i + d_i$ doesn't change, since $x_i$ is updated to $x_i + \delta_i$
and $d_i$ is updated to $d_i - \delta_i$. Therefore if $\tilde{x},\tilde{d}$
are the values of allocation and demands the next time point \cp is
reached, then: $f_{\tilde{x}+\tilde{d}}([n]) =
f_{\tilde{x}+\tilde{d}_{-\hi}}([n]) = f_{x+(0,d_{-\hi})}([n]) = f([n])$.
\end{proofof}

\begin{proofof}{Lemma \ref{lemma:self-saturation}}
Notice it is enough to show that $f_{x+(0,d_{-i})}([n]) = f_{x+(0,d_{-i})}([n]
\setminus i) + x_i$ since one can extend for every set $S \ni i$ using
submodularity, after all: $x_i \geq f_{x+(0,d_{-i})}(S) - f_{x+(0,d_{-i})}(S
\setminus i) \geq f_{x+(0,d_{-i})}([n]) - f_{x+(0,d_{-i})}([n] \setminus i) =
x_i$.

When point \cp is reached for the first time in the algorithm, the
invariant $f_{x+(0,d_{-i})}([n]) = f_{x+(0,d_{-i})}([n] \setminus i) + x_i$
holds trivially since $x_i = f([n]) - f([n] \setminus i) = f_{x+(0,d_{-i})}([n])
-  f_{x+(0,d_{-i})}([n] \setminus i)$, by Lemma \ref{lemma:all_goods_sold} and
the fact that at this point $d_i = \infty$ for all $i$.

Now we will show that the invariant is preserved between two consecutive
visits to point \cp. Say that $x,d$ is the allocation and demands in
point \cp just before we increase price $p_\hi$. Let $d'_\hi$ be the new
demand for $\hi$ and $d' = (d'_\hi, d_{-\hi})$. We want to show that:
$f_{x+(0,d'_{-i})}([n]) = f_{x+(0,d'_{-i})}([n] \setminus i) + x_i$. This is
trivial for $i = \hi$, since the expression is unaffected. For $i \neq \hi$, we
can use Lemma \ref{lemma:auxiliary}:
$$\begin{aligned} & f_{x+(0,d'_{-i})}([n])-  f_{x+(0,d'_{-i})}([n] \setminus i)
= \min\{f_{x+(0,d_{-i})}([n]), f_{x+(0,d_{-i})}([n] \setminus \hi) + x_\hi +
d'_\hi \} - \\ & \qquad \min\{f_{x+(0,d_{-i})}([n] \setminus i),
f_{x+(0,d_{-i})}([n] \setminus \{i,\hi\}) + x_\hi + d'_\hi \} \end{aligned} $$
Now we can analyze four cases based on which expression achieves the
minimum. One of those four cases is impossible, since by
submodularity: $f_{x+(0,d_{-i})}([n] \setminus \{i,\hi\}) - f_{x+(0,d_{-i})}([n]
\setminus i) \geq f_{x+(0,d_{-i})}([n]) - f_{x+(0,d_{-i})}([n] \setminus \hi)
$, the minimum can't be achived by $f_{x+(0,d_{-i})}([n])$ in the first and by
$f_{x+(0,d_{-i})}([n] \setminus \hi) + x_\hi + d'_\hi$ in the second. Now, we
proceed by analyzing the remaining cases:
\begin{itemize}
 \item first / first: $f_{x+(0,d'_{-i})}([n])-  f_{x+(0,d'_{-i})}([n] \setminus
i) = f_{x+(0,d_{-i})}([n] ) - f_{x+(0,d_{-i})}([n] \setminus i) = x_i$ by the
invariant.
 \item second / second: $f_{x+(0,d'_{-i})}([n])-  f_{x+(0,d'_{-i})}([n]
\setminus i) =  [ f_{x+(0,d_{-i})}([n] \setminus \hi) + x_\hi +
d'_\hi ]  - [f_{x+(0,d_{-i})}([n] \setminus \hi) + x_\hi + d'_\hi] \geq
f_{x+(0,d_{-i})}([n] ) - f_{x+(0,d_{-i})}([n] \setminus i) = x_i$. The
inequality in the other direction is trivial.
 \item second / first: in this case by the fact that the minimum is achieved by
the first term in the second expression, $f_{x+(0,d_{-i})}([n] \setminus i) \leq
f_{x+(0,d_{-i})}([n] \setminus \{i,\hi\}) + x_\hi + d'_\hi$, therefore:
$f_{x+(0,d'_{-i})}([n])-  f_{x+(0,d'_{-i})}([n]
\setminus i) =  f_{x+(0,d_{-i})}([n] \setminus \hi) + x_\hi +
d'_\hi - f_{x+(0,d_{-i})}([n] \setminus i) \geq f_{x+(0,d_{-i}}([n] \setminus
\hi) - f_{x+(0,d_{-i}}([n] \setminus \{i, \hi\}) \geq x_i$. The
inequality in the other direction is trivial.

This shows that the invariant holds even after we decrease the $d_\hi$. Before
we reach point \cp again, allocation and demands change because of
clinching. This, however, doesn't change the invariant, since clinching adds an
extra $\delta_i$ amount to $x_i$ but subtracts the same amount from $d'_i$, so
$x+d$ is constant. Therefore, clinching doesn't affect the invariant.
\end{itemize}

\end{proofof}

%\input{online-clinching-missing-proofs.tex}
%\input{online-clinching-repeated-values.tex}
%\input{online-clinching-explicit.tex}

% That's all folks!
\end{document}